\documentclass{article}

\usepackage{fullpage}

\usepackage{algorithmic}
\usepackage{algorithm}
\usepackage{url}
\usepackage{graphicx}
\usepackage{color}
\usepackage{listings}
\usepackage{subfigure}
\usepackage{authblk}

\usepackage{amsthm}
\usepackage{amsmath}
\usepackage{amssymb}
\newcommand{\reals}{\mathbb{R}}
\newcommand{\nats}{\mathbb{N}}
\DeclareMathOperator{\mf}{\enspace .}
\DeclareMathOperator{\mc}{\enspace ,}
\newtheorem{lemma}{Lemma}

\newtheorem{theorem}{Theorem}
\newtheorem{corollary}{Corollary}
\newtheorem{conjecture}{Conjecture}
\newtheorem{remark}{Remark}
\newcommand{\rvar}{r}

\floatname{algorithm}{Listing}

\title{Data Structures for Task-based Priority Scheduling\thanks{A
    short poster summary of this paper was presented at the 19th ACM
    PPoPP 2014 conference~\cite{ppopp14}.}  }

\author[1]{Martin Wimmer}
\author[2]{Daniel Cederman}
\author[1]{Francesco Versaci}
\author[1]{Jesper Larsson Tr{\"a}ff}
\author[2]{Philippas Tsigas}
\affil[1]{Faculty of Informatics\\Vienna University of Technology\\
Favoritenstrasse 16, 1040 Vienna, Austria\\ \texttt{\{wimmer,versaci,traff\}@par.tuwien.ac.at}}
\affil[2]{Computer Science and Engineering\\Chalmers University of Technology\\
412 96 G\"oteborg, Sweden\\ \texttt{\{cederman,tsigas\}@chalmers.se}}

\begin{document}

\maketitle

\begin{abstract}
Many task-parallel applications can benefit from attempting to execute
tasks in a specific order, as for instance indicated by priorities
associated with the tasks.  We present three lock-free data structures
for priority scheduling with different trade-offs on scalability and
ordering guarantees. First we propose a basic extension to
work-stealing that provides good scalability, but cannot provide any
guarantees for task-ordering in-between threads. Next, we present a
centralized priority data structure based on $k$-fifo queues, which
provides strong (but still relaxed with regard to a sequential
specification) guarantees. The parameter $k$ allows to dynamically
configure the trade-off between scalability and the required ordering
guarantee. Third, and finally, we combine both data structures into a
hybrid, $k$-priority data structure, which provides scalability
similar to the work-stealing based approach for larger $k$, while
giving strong ordering guarantees for smaller $k$. We argue for using
the hybrid data structure as the best compromise for generic,
priority-based task-scheduling.

We analyze the behavior and trade-offs of our data structures in the
context of a simple parallelization of Dijkstra's single-source
shortest path algorithm. Our theoretical analysis and simulations show
that both the centralized and the hybrid $k$-priority based data
structures can give strong guarantees on the useful work performed by
the parallel Dijkstra algorithm.  We support our results with
experimental evidence on an 80-core Intel Xeon system.
\end{abstract}


\section{Introduction}

Parallel tasks is a convenient parallel programming pattern for
exposing independent work-units that can be scheduled over multiple
processing elements.  The popular work-stealing
paradigm~\cite{BlumofeLeiserson99} is an efficient way to schedule
such parallel work-loads of independent tasks, and forms the basis for
well-known frameworks such as Cilk\verb!++!~\cite{Leiserson10}, Intel
Threading Building Blocks (TBB)~\cite{KukanovVoss07} and
X10~\cite{Charles05}.  Some task-parallel systems, like TBB and
StarPU~\cite{Augonnet10}, support assigning priorities to tasks to
influence the task execution order, with priorities typically
restricted to a small number of discrete values.  Some applications
that rely on priority
scheduling~\cite{QuintanaOrtivandeGeijnZeeChan09} resort to their own
centralized scheduling scheme, based on a shared priority queue.
However, it can be argued that shared priority queues are not
necessarily a good solution for the priority scheduling
problem~\cite{pingali2011}. Other schemes, that rely on decentralized
priority queues, cannot provide any guarantees on the execution order
of tasks in-between different
threads~\cite{WimmerCedermanTraffTsigas13:strategies,ppopp13}.

In this work we present three designs of lock-free data structures for
priority scheduling, each with different trade-offs concerning
scalability and scheduling guarantees.  The designs include a priority
work-stealing data structure, a centralized data structure inspired by
$k$-fifo queues~\cite{kfifo} with $k$-relaxed semantics, as introduced
by Afek et al.~\cite{AfekKorlandYanovsky10}, and a hybrid data
structure combining both ideas.  The designs support choosing the
value of $k$ per task, allowing kernels with different ordering
requirements to coexecute. Using the single-source shortest path
problem as an example, we show how the different approaches affect the
prioritization and show how bounds on the number of examined nodes can
be given. We argue that priority task scheduling allows for an
intuitive and easy way to parallelize the otherwise hard to
efficiently parallelize single-source shortest path
problem. Experimental evidence supports the good scalability of the
resulting algorithm.

The larger aim of this work is to understand trade-offs between
priority guarantees and scalability in task scheduling systems. We
show that $\rho$-relaxation is a valuable technique for improving
scalability while still providing semantic guarantees. The lock-free,
hybrid $k$-priority data structure shows that data structures can be
implemented that have scalability on par with work-stealing, while at
the same time providing strong priority scheduling guarantees,
depending on the value used for $k$. Our theoretical results open up
possibilities for even more scalable data structures due to further
relaxations that do not influence the bounds.  A \verb!C++!
implementation of our data structures and applications is available
for download as part of the open source task-scheduling framework
Pheet~\cite{WimmerCedermanTraffTsigas13:strategies,ppopp13,mtaap13}\footnote{\url{http://www.pheet.org}}.

\section{The model}

The focus of this work is the presentation and evaluation of data
structures for task scheduling with priorities. The model used is the
async-finish model, which is well-known from X10~\cite{Charles05} and
other task-based programming models, where new tasks can be spawned
throughout the execution of a task. Tasks can be synchronized using
finish regions. A finish region is a blocking synchronization
primitive, where execution can only continue after all tasks
transitively spawned inside the finish region have been executed.

We extend the task model to support priority scheduling. Our model of
priority scheduling relies on a comparison operator between tasks,
which can be specific to an application/algorithm. We call the
comparison operator the \emph{priority function} for the remainder of
the paper. Our framework allows the programmer to store
application-specific information alongside a task, to be used in the
priority function. As an example, in the single-source shortest path
application used for the evaluation in Section~\ref{sec:eval}, each
task represents a single node relaxation. The priority function for
this application uses the length of the shortest path found so far for
each node, and prioritizes nodes with smaller distance values, similar
to Dijkstra's algorithm. The model was described in our previous
work~\cite{WimmerCedermanTraffTsigas13:strategies}, which discusses
programmability aspects with other example applications.

Our scheduling system relies on \emph{help-first}
scheduling~\cite{GuoBarik09}, where newly spawned tasks are stored for
later execution by any thread, and the current thread proceeds with
the continuation.  This can be contrasted with \emph{work-first}
scheduling, where the continuation is stored for later execution and
the newly spawned task is executed by the current thread.  Work-first
scheduling has better space bounds for general task scheduling, but it
is not feasible for priority scheduling since it relies on a fixed
order in which tasks are executed (depth-first). Instead, a priority
function has to be chosen for tasks that gives bounds on the number of
concurrently available tasks. For many applications, like the
single-source shortest path application used in this paper, the
intuitive prioritization scheme inherently has bounds on the number of
concurrently available tasks.

In our model, the task scheduling system has multiple threads of
execution, each with its own supporting data structures. We use the
term \emph{place} to denote a single thread of execution and its
supporting local data structures. Whenever a place is idle, it
retrieves a task that has been stored for later execution and executes
it until it is finished. The scheduling system terminates when all
tasks have finished executing and no new tasks were created.

\subsection{Data structure model}

In this paper we discuss three different approaches for how a priority
data structure for storing tasks in our model can be implemented. All
three data structures rely on the same interface to interact with the
scheduling system. Each data structure consists of a
\emph{centralized}, global component that is shared by all places, and
which is accessed by every place in the same manner. In addition, each
place stores a separate, \emph{local} component of the data
structure. This allows for asymmetric access schemes, where the
\emph{owner} of the local component (the thread associated with the
place) is the only thread allowed to perform specific operations,
thereby allowing for simpler synchronization schemes.

The scheduling system interacts with the data structure using two
functions, \texttt{push} and \texttt{pop}. Both functions are executed
in the context of a specific place, therefore giving access to the
local component of the priority data structure for the given
place. The function \texttt{push} is called whenever a new task is
spawned, and stores the new task in the data structure for later
execution. The function \texttt{pop} returns a task and deletes it
from the data structure. Each task that has been added to the data
structure will be returned by \texttt{pop} exactly once. We allow
\texttt{pop} to spuriously fail as long as another thread is making
progress. The task returned by \texttt{pop} does not necessarily have
to be the highest priority task. The guarantees on the ordering of
tasks provided are specific to the data structure implementation and
are discussed in Section~\ref{sec:eval}.

\subsection{$\rho$-relaxation}
\label{sec:k-relax}

In order to improve the scalability of the proposed data structures,
we adopt a $\rho$-relaxation scheme, as introduced by Afek et
al.~\cite{AfekKorlandYanovsky10}, which is a temporal property that
allows certain items in the data structure to be \emph{ignored}. We
say an item is ignored whenever an item of lower priority is returned
by a \texttt{pop} operation.

The centralized $k$-priority data structure presented in this work
satisfies $\rho$-relaxation in the following sense: a \texttt{pop}
operation is allowed to ignore the last $k$ items added to the data
structure, which, in the worst case, might be the top $\rho=k$ by
priority.

On the other hand, \texttt{pop} operations for the hybrid $k$-priority
data structure are allowed to ignore the last $k$ items added by each
thread, which implies that, being $P$ the number of threads, up to
$\rho=Pk$ items might be ignored in total.

\section{Data structures}
\label{sec:ds}

In this section we give a high level description of the priority data
structures for task scheduling compared in this paper.

\subsection{Work-stealing}

One approach that we evaluate is to adapt work-stealing to priority
scheduling, by using priority queues instead of standard deques,
similar to an approach presented in previous work of the
authors~\cite{WimmerCedermanTraffTsigas13:strategies,ppopp13}. This
preserves the scalability of work-stealing, while imposing local
prioritization on tasks. Due to the decentralized nature of
work-stealing, where each thread is only aware of its own tasks, no
global priority ordering can be established. Therefore, no guarantee
can be given on the priority of tasks that are being executed.

Our implementation of work-stealing uses a local priority queue per
place, which is used by the \texttt{push} and \texttt{pop} operations
to store and prioritize tasks.  When the \texttt{pop} operation is
called on a place where the priority queue is empty, it chooses a
random place and steals half the tasks from that place's priority
queue. Stealing half the tasks allows tasks that are generated at one
place to quickly spread throughout the system~\cite{Hendler02}.

\subsection{Centralized $k$-priority data structure}
\label{sec:centralized}

The straightforward way to maintain strong guarantees on the priority
of tasks is to use a data structure with the semantics of a
centralized, global priority queue. Each \texttt{push} and
\texttt{pop} operation directly communicates with the priority
queue. It has been shown that a centralized, global priority queue
exhibits a lot of congestion when used in a scheduling
system~\cite{pingali2011}, since all threads try to access the highest
priority task. To reduce congestion, we use a $\rho$-relaxation scheme
as described in Section~\ref{sec:k-relax}.

\subsection{Hybrid $k$-priority data structure}

The hybrid $k$-priority data structure combines the work-stealing and
the $\rho$-relaxation ideas into a single data structure. The main
idea is that each place maintains its own, local priority queue, and
that synchronization is only performed if either a place runs out of
work, or if the guarantees provided by $\rho$-relaxation are violated.

\section{Implementations}

In this section we provide the implementations of both the
$k$-priority data structures. Since the $k$-priority data structures
constitute the main contribution of the paper, we omit the details of
the work-stealing data structure.

\subsection{Centralized $k$-priority data structure}
\begin{figure}[t]
\centering
\includegraphics{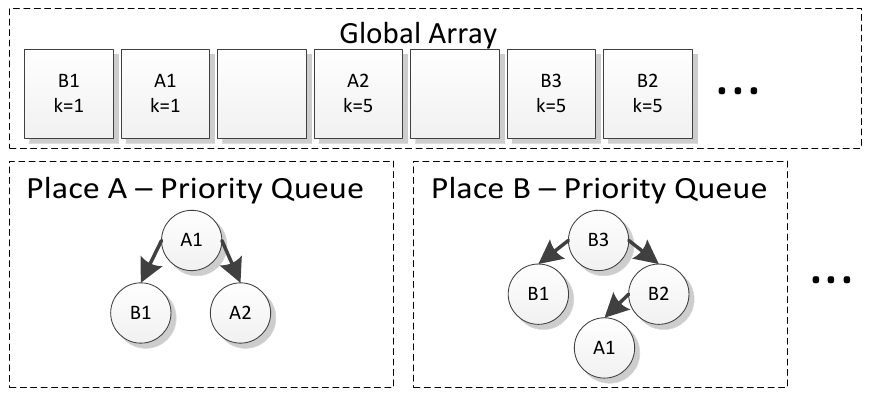}
\caption{Centralized $k$-priority data structure. Each place maintains
  its own priority queue with references to items in the global
  array. The newest (rightmost) items in the global array are only
  visible to the place that created them.}
\label{fig:centralkdesign}
\end{figure}

The basic idea of the centralized $k$-priority data structure is to
create a global priority ordering between all the tasks available in
the system, while allowing each thread to miss up to $k$ of the newest
tasks, as long as each task is seen by at least one thread. To achieve
such a $\rho$-relaxation, we split the data structure into two
components. One component, a global, shared array, is used to share
tasks between all threads and to maintain information about which
tasks must be globally visible so as not to violate the
$k$-requirement.  Randomization is used to improve scalability when
adding elements to the global array. The other component, which
consists of local priority queues for each place (thread), is used to
maintain the priority ordering of the tasks visible to each
place. This is depicted in Figure~\ref{fig:centralkdesign}. Any
sequential implementation of a priority queue can be used for the
local priority queues, since each priority queue is only accessed in
the context of a single place, and therefore only by a single thread.

\subsubsection{The \texttt{push} operation}

The \texttt{push} operation stores the task, together with some
additional information, in a structure which we call an
\emph{item}. For an item to be visible to all threads, it needs to be
added to the global array.  The items in the global array are stored
in an order close to sequential. A task may only be placed up to $k$
positions away from its correct sequentially consistent position.

Pseudocode for the \texttt{push} operation is shown in
Listing~\ref{lst:centralk_push}. There, we choose a random position in
the range from \texttt{tail} to $\texttt{tail} + k$ and try to put the
item into the array at the chosen position, if the position has not
yet been taken by another item. In case the position is taken, a
linear search is performed inside the \texttt{tail} to $\texttt{tail}
+ k$ range until a free position is found or all positions have been
checked. If all positions are taken, \texttt{tail} can be updated to
$\texttt{tail} + k$ and the search restarted. This scheme for adding
items to an array was inspired by the $k$-fifo queues of Kirsch et
al.~\cite{kfifo}.

As soon as the item has been added to the global array, a reference to
it is added to the priority queue of the place at which it was
created. This guarantees that at least one thread will attempt to
execute this task next, if it has the globally highest priority.

\begin{algorithm}[t]
\begin{lstlisting}[mathescape=true,columns=flexible,escapechar=!]
void push(Place place, int k, Task task) {
	Item it = new Item(place, k, task);	
	
	// Attempt until successful
	while(true) {
		int t = tail;
		
		// Choose a random offset at which to put item
		int offset = rand(0, k - 1);
	
		// try all indices in k-range starting at offset
		for(int i = offset; i < offset + k; ++i) {
			int pos = t + (i % k);
			// A tag of -1 refers to a taken item. We store pos
			// in the tag field to omit the ABA problem
			it.tag = pos;
			// Try to put item into global array
			if(CAS(global_array[pos], null, it)) {
				// Item was succesfully put into array
				// Now put a reference into local priority queue
				ItemRef ref = new ItemRef(pos, it);
				place.prio_queue.push(ref);				
				return;
		}}
		
		// No more free slot found, try updating tail
		// One thread will succeed, no need for checking which
		CAS(tail, t, t + k);
}}
\end{lstlisting}
\caption{Pseudocode for \texttt{push} in the centralized $k$-priority data structure.}
\label{lst:centralk_push}
\end{algorithm}

\subsubsection{The \texttt{pop} operation}
\begin{algorithm}[t]
\begin{lstlisting}[mathescape=true,columns=flexible,escapechar=!]
Task pop(Place place) {
	// Check for new tasks in global array
	while(place.head < tail) {
		if(global_array[place.head].place $\neq$ place) {
			ItemRef ref = new ItemRef(place.head, global_array[place.head]);
			place.prio_queue.push(ref);
	}}
	
	ItemRef ref;
	while(ref = place.prio_queue.pop()) {
		Task task = ref.it.task;
		// Take item atomically by setting tag to -1
		if(CAS(ref.it.tag, ref.tag, -1)) {
			// Success, return task
			return task;
		}
		// Recheck for new tasks in global array again
		... //  (not shown)
	}
	
	// Priority queue is empty, try to find random task
	int offset = rand(0, $k_{max}$ - 1);
	if(global_array[tail + offset] $\neq$ null &&
		global_array[tail + offset].k $\leq$ offset) {
		Item it = global_array[tail + offset];
		Task task = it.task
		// Take item atomically by setting tag to -1
		if(CAS(it.tag, tail + offset, -1))
			return task;
	}
	return null;
}
\end{lstlisting}
\caption{Pseudocode for \texttt{pop} in the centralized $k$-priority data structure.}
\label{lst:centralk_pop}
\end{algorithm}

Pseudocode for the \texttt{pop} operation can be found in
Listing~\ref{lst:centralk_pop}. The \texttt{pop} operation checks
whether \texttt{tail} has changed since the last time it was checked,
and if so adds all the newly added tasks to the local priority
queue. Each place maintains its own \texttt{head} index into the
global array, to track which items have already been seen. Tasks that
have been created by the same place can be omitted, since they were
already added to the priority queue at the push operation. Next, the
highest priority task is removed from the priority queue, and an
attempt is made to mark the task as taken, by atomically setting the
\texttt{tag} of the item to $-1$ using a \emph{compare-and-swap}
operation (CAS). Only one thread can succeed in updating the tag. In
case of failure, the global array is rechecked for new tasks before
trying again.

If the priority queue is empty, there can be up to $k$ tasks stored
after \texttt{tail} waiting for their execution, stored in the
\texttt{tail} of the global array and its $k$ subsequent
positions. Our data structure allows for varying values for $k$ per
task, thereby making it necessary to specify a maximum value for
$k$. We chose $k_{\max} = 512$ for our implementation. Since there are
at most $k$ tasks stored after \texttt{tail}, no priority ordering
needs to be guaranteed if there are no tasks before the \texttt{tail},
and a random position can be checked for a task to execute. If a task
is found, the $k$ value stored with the task is rechecked, to make
sure not more than $k$ tasks are ignored (\texttt{tail} might have
been updated in the meantime). Since we allow for spurious failures on
\texttt{pop} as long as another thread is making progress (executing a
task), it is not necessary to exhaustively search for all tasks stored
after tail, a random attempt suffices.

\subsubsection{Additional implementation details} % omitted so far
\label{sec:centralk_details}

So far we have assumed that the global array used for storing tasks is
unbounded. In practice, we implemented the global array as a linked
list of arrays. Whenever an index is requested that is outside the
bounds of the existing arrays, a new array is allocated and added to
the end of the linked list using a single \emph{compare-and-swap}
operation.

Each array in the linked list can be deleted as soon as all tasks
stored in the array have been executed and the \texttt{head} indices
of all places point to positions in arrays that are successors of the
given array. The first condition can be lazily checked using a garbage
collection scheme. We use the wait-free garbage collection scheme by
Wimmer~\cite{Wimmer13} for these purposes. The second condition can be
checked by atomically decrementing a reference counter whenever a head
index moves on to the next array. If the reference counter was
initialized to the number of places in the beginning, it is guaranteed
that no place will scan the array for new tasks once the counter
reaches zero.

It is also necessary to clean up all the \emph{items} used for storing
tasks. For performance reasons we decided on a reuse scheme, where an
item can be reused for a new task as soon as the previous task has
been executed. The use of a tag for each item, which is initialized to
the item's position in the global array, guards against the
ABA-problem, since positions for items are strictly increasing. Also,
since items may be reused directly after the compare-and-swap, the
task has to be read out of the item before the compare-and-swap.

Both the \texttt{head} and \texttt{tail} indices in the data structure
are strictly growing, therefore it is necessary to take a possible
wraparound into account. We use 64-bit values, which ensures that
wraparounds will only occur after a long time. Due to the long
timespan between wraparounds, we consider it unlikely that an ABA
problem will occur due to colliding indices.

\subsubsection{Correctness}

In this section we argue that the centralized $k$-priority data
structure is lock-free and linearizable.

\begin{theorem}
\label{theorem:centralizedpush}
The \texttt{push} operation is lock-free and linearizable.
\end{theorem}

\begin{proof}
All memory allocation is done using a wait-free memory
manager~\cite{Wimmer13}. The \texttt{push} operation tries to find an
empty slot in the global array to insert its item. It searches $k$
positions from a local copy of the tail index. If no empty slot is
found, the tail is moved forward at least one step, either by the
current thread or a concurrent thread. This is repeated until an empty
slot is encountered.  If no empty slot is found, or the CAS used to
insert the item fails, another thread must have succeeded in inserting
at least one item.  This is in accordance with the lock-free property.

The \texttt{tail} can only be moved when all the slots before its new
position are filled. So when an item is inserted, it is guaranteed to
be at most $k$ steps from the \texttt{tail}, since the \texttt{tail}
cannot be moved before the successful insertion.  Push operations can
thus be linearized relative to each other at the point where they
manage to insert their items into the array. The actual value of
\texttt{tail} might differ from a sequential execution, but this does
not affect the semantics.
\end{proof}

\begin{theorem}
\label{theorem:centralizedpop}
The \texttt{pop} operation is lock-free and linearizable.
\end{theorem}

\begin{proof}
The \texttt{pop} operation has to check the global array for new
items.  This can be done in a bounded number of steps if no other
thread is making progress or adds new items. After reading the global
array, the operation tries to acquire one of the tasks referenced in
the priority queue. The size of the priority queue only grows when
another thread is making progress and adds new items to the global
array.

A \texttt{push} operation that inserts an item must be linearized
before the \texttt{pop} operation that reads it.  If the item is not
in the $k$-relaxed part after the tail, the \texttt{pop} operation
also needs to be linearized after the \texttt{push} operation that
updated the \texttt{tail} to the value seen by the \texttt{pop}
operation. Relative to each other, \texttt{pop} operations should be
linearized at the point where they take the item using CAS.

All \texttt{pop} operations that observe a certain \texttt{tail} are
linearized after the update of \texttt{tail} that stored the observed
value and before the next. Only items after the \texttt{tail} will be
ignored by any thread. At most $k$ additional items can be stored
after \texttt{tail} and therefore no more than $k$ items will be
ignored, regardless of the order in which items are popped from the
data structure.
\end{proof}

\subsection{Hybrid $k$-priority data structure}
\label{sec:k-prio-impl}
\begin{figure}[t]
\centering
\includegraphics{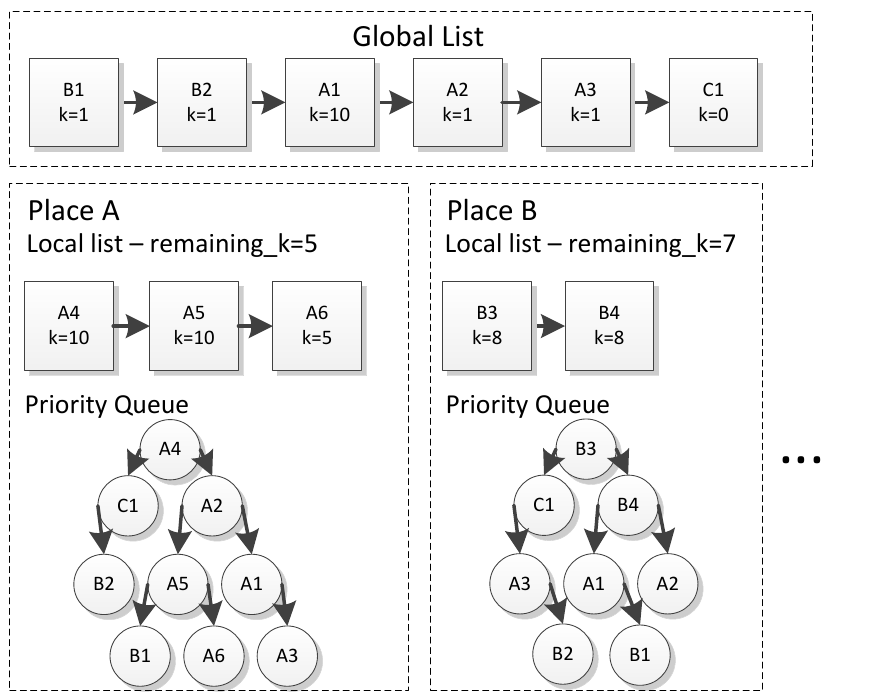}
\caption{Hybrid $k$-priority data structure. Each place maintains its
  own priority queue with references to items. Each place adds new
  items to its local list as long as the $\rho$-relaxation guarantees
  are not violated. If adding a new item would violate these
  guarantees, the local list is appended to the global list, and a new
  local list is created.}
\label{fig:hybridkdesign}
\end{figure}

The hybrid $k$-priority data structure consists of three components:
(a) a global list storing tasks visible to all places, (b) one local
task list per place, containing up to $k$ tasks that are not
guaranteed to be visible to all places, and (c) one priority queue per
place storing references to tasks in the global and local lists,
ordered by priority. After more than $k$ tasks have been added to the
local task list, the place makes all local tasks globally visible by
moving them to the global task list. A task can be referenced by
multiple priority queues at the same time, which is required to
guarantee that at most the $k$ newest tasks by each thread are missed.

\subsubsection{The \texttt{push} operation}
\begin{algorithm}[t]
\begin{lstlisting}[mathescape=true,columns=flexible,escapechar=!]
void push(Place place, int k, Task task) {
	Item it = new Item(place, priority, task);
	// Place task in local list and priority queue
	place.local_list.add(it);
	place.prio_queue.push(new ItemRef(it));

	// All items need to be made globally visibly
	// to not violate the $\rho$-relaxation requirement
	remaining_k = min(remaining_k - 1, k);
	if(remaining_k == 0) {
		// Add local list to global list
		do { 
			processGlobalList(place) 
		} while($\lnot$CAS(global_list.tail.next, null, local_list.head));
		// Create a new local list
		place.local_list = new List();
		remaining_k = !$\infty$!;
}}

// Add references to unread items from
// the global list to the local priority queue
void processGlobalList(Place place) {
	while(place.iterator$\neq$global_list.tail) {
		Item it = place.iterator.item()
		// Do not add local or already taken tasks
		if(it.place $\neq$ place and !$\lnot$!it.taken)
			place.prio_queue.push(new ItemRef(it));
		place.iterator = place.iterator.next;
}}
\end{lstlisting}
\caption{Pseudocode for \texttt{push} in the hybrid $k$-priority data
  structure.}
\label{lst:hybridkpush}
\end{algorithm}

The \texttt{push} operation (see Listing~\ref{lst:hybridkpush}) adds a
new task into the data structure. Each task is associated with a
specific value for $k$, which determines how many tasks are allowed to
be added before the task has to be made public. The semantics of $k$
are that no more than $k$ tasks are allowed to be added to the local
list of tasks before the given task must be published.

The \texttt{push} operation proceeds as follows: first, the task is
inserted into the local list of the given place, and a reference is
stored in the local priority queue of the place. Afterwards, a check
is performed whether more tasks can be added without needing to
publish any of the locally stored tasks, which is the case if the
variable \texttt{remaining\_k} is greater than zero. If any of the
locally stored tasks needs to be made available globally, the local
list of tasks is appended to the global list. A new, empty local list
is then created, which will be used in the next push operations.

\subsubsection{The \texttt{pop} operation}

\begin{algorithm}[t]
\begin{lstlisting}[mathescape=true,columns=flexible,escapechar=!]
Task pop(Place place) {
	do {
		processGlobalList(place);
		// Try to take the highest priority task
		while($\neg$place.prio_queue.empty()) {
			Ref r = prio_queue.pop();
			if(!$\lnot$!r.item.taken) {
				Task ret = r.item.task;
				if(TAS(r.item.taken))
					return ret;
			}
			processGlobalList(place);
		}
		
		// If the priority queue is empty, add references
		// to remote tasks from a pseudo-random place
		List vl = getRandVictim().local_list;
		foreach(Item it in vl) {
			if(it.place $\neq$ place and !$\lnot$!it.taken)
				place.prio_queue.push(new ItemRef(it));
		}
	} while($\neg$place.prio_queue.empty());
	return null;
}
\end{lstlisting}
%\end{multicols}
\caption{Pseudocode for \texttt{pop} in the hybrid $k$-priority data structure.}
\label{lst:hybridkpop}
\end{algorithm}

The \texttt{pop} operation (see Listing~\ref{lst:hybridkpop}) pops a
reference to the highest-priority task from the local priority queue
and tries to mark the task as taken by setting the \texttt{taken} flag
with an atomic \emph{test-and-set} operation. If it succeeds the task
is returned.  To make sure that no more than $k$ tasks per place are
ignored, the local priority queue has to be regularly updated with the
newest additions to the global list.  This is always done before a
task is popped from the priority queue.

If the priority queue is empty after processing the global list, an
attempt is made to find tasks stored locally at another place. This is
called \emph{spying}. Spying is related to stealing in work-stealing
systems in that a remote place is selected semi-randomly among all
places and searched for tasks that have not yet been executed. The
main difference is that tasks that are encountered during spying are
not removed from the owner's local list of tasks. Instead, only
references to the given tasks are stored in the priority queue. This
is necessary to avoid breaking $\rho$-relaxation guarantees, but also
greatly simplifies synchronization.

Spying is only required when less than $k$ tasks are in flight at each
place, so that $k$-prioritization is not violated if not all the
victim's tasks are encountered during spying. Therefore the semantics
of spy allows spurious failures, where the spying place does not see
all the victim's tasks. This allows for a very lightweight
synchronization scheme for spying. Spying may lead to tasks appearing
in a single priority queue twice, but never more often, if the task
was first encountered during spying, and was later made available
globally. This does not affect the correctness however, since a task
can only be executed once.

\subsubsection{Additional implementation details} 

For efficiency reasons, our implementation of the hybrid $k$-priority
data structure does not use linked lists, but instead uses a linked
list of arrays, which can be implemented in a similar manner as for
the centralized $k$-priority data structure, as described in
Section~\ref{sec:centralk_details}.

Similarly, the memory management for \emph{items}, which are used to
store all the information about a task, can be taken over from the
centralized data structure. Alternatively, items can be stored
in-place in the linked list of arrays for higher efficiency. To guard
against the ABA problem when an item is reused, we use a tag instead
of an atomic flag to mark an item as taken, similar to the tag used in
the centralized $k$-priority data structure. Since, contrary to the
centralized data structure, an item has no global index at the time it
is added to the data structure, each place maintains its own local
indices, which are used to fill the tag field. An offset is stored in
each array in the linked list before it is linked to the global list,
to allow other threads to calculate the offsets of the items.

Spying does not put any tasks into the local task list of the spying
place (contrary to steal-half work-stealing), which makes the given
place appear as being out of work for other spying places. To ensure a
proper distribution of tasks throughout the whole system, each place
stores a reference to its last successful spying victim. In case a
victim is encountered with no local work, its last successful spying
victim is checked instead.

\subsubsection{Correctness}

In this section we argue that the hybrid $k$-priority data structure
is lock-free and linearizable.

\begin{theorem}
\label{theorem:push}
The \texttt{push} operation is lock-free and linearizable.
\end{theorem}

\begin{proof}
All memory allocation is done using a wait-free memory
manager~\cite{Wimmer13}. When there are less than $k$ tasks in the
local list, the entire push operation is done locally and is thus
wait-free. When the local list has $k$ tasks, it is added to the
global list. This step requires making sure that the entire global
list has been read and then adding the local list to the end.  Adding
the local list to the global list can fail if another place adds its
list first, but this means another place made progress. Reading and
adding to the global list is thus lock-free.

A push operation for a task which is taken before being added to the
global list, has its linearization point where it is added to the
local list. Before this point the task is not visible to any remote
place, while after the point it can be spied and taken by any place.
For tasks which are taken after being globally announced, the push
operation is linearized at the point where the local list was
atomically added to the global list. The pop operation always reads
all new tasks when reading from the global list, so push operations
can be linearized in any order relative to each other.
\end{proof}

\begin{theorem}
\label{theorem:pop}
The \texttt{pop} operation is lock-free and linearizable.
\end{theorem}

\begin{proof}
At certain points the pop operation needs to make sure it has read the
entire global list.  The global list can only grow if another place is
making progress, which makes reading the list lock-free. Multiple
places may try to acquire the same task, but only one will
successfully take it. The number of already taken tasks can only grow
if another place is making progress.  If the priority queue is empty
and the global list has been read, an attempt is made to spy on the
local list of another place. The length of the remote local list is
bounded by $k$, making the spying wait-free.

A successful pop operation has to be linearized relative to other pop
operations at the point where the task was atomically marked as
taken. Relative to push operations, the pop has to be linearized at
the point where the global list was last read. At this point the place
has a snapshot that contains all but at most $(P-1)k$ tasks, where $P$
is the number of places. If instead the pop was linearized when the
task was taken, another push operation could have added new tasks to
the global list before that, causing the place to miss more than the
allowed $Pk$ tasks.

An unsuccessful pop operation is linearized at the point where the
global list was last read. At this point both the global and the local
lists were empty.  The final spying part is allowed to fail even
though there might be tasks at other places.
\end{proof}

\section{Evaluation}
\label{sec:eval}

The goal of this section is to show how $\rho$-relaxation can help to
give bounds beyond what can be achieved with work-stealing for the
execution of a parallel application. For this we base our evaluation
on the well-understood single-source shortest path problem. We derive
theoretical bounds, which are also later verified through
simulation. Furthermore, we show that the bounds are applicable to all
variations of $\rho$-relaxation presented in this paper. In addition
we provide experiments that show the practical performance gains of
our approach.

\begin{algorithm}[t]

\begin{lstlisting}[mathescape=true,columns=flexible,escapechar=!]
void relaxNode(Graph graph, int node, int distance) {
	int d = graph[node].distance;
	if(d !$\neq$! distance) {
		// Dead task, distance has already been improved in the meantime
		return;
	}
	
	for(int i = 0; i < graph[node].num_edges; ++i) {
		int new_d = d + graph[node].edges[i].weight;
		int target = graph[node].edges[i].target;
		int old_d = graph[target].distance;
		
		// Check if path through this node is shorter
		while(old_d > new_d) {
			// Try to update distance value
			if(CAS(&(graph[target].distance), old_d, new_d)) {
				spawn(new_d, // priority, smaller is better
					relaxNode, // Function used for task
					graph, target, new_d)); 					
				break;
			}
			old_d = graph[target].distance;
}}}
\end{lstlisting}

\caption{Pseudocode for a single node relaxation in our parallel single-source shortest path algorithm.}
\label{lst:sssp}

\end{algorithm}

\subsection{Application}

We base our evaluation on a simple and well-understood example
application that profits from priorities, and consider the
single-source shortest path problem
(SSSP)~\cite{AhujaMagnantiOrlin93}. We focus on a simple
parallelization of Dijkstra's algorithm. Dijkstra's algorithm
maintains a tentative distance value for each node in the graph. At
each iteration, a node relaxation is performed, where the tentative
distance values of the neighboring nodes are decreased if the path
through the relaxed node is shorter. At termination, the distance from
the source node is available for each node. A priority queue is used
to decide the order in which nodes are relaxed; the priority ordering
guarantees that each node is relaxed exactly once.

Our parallel version relaxes multiple nodes in parallel. Due to the
parallelization some node relaxations might be performed prematurely,
when a node is not yet \emph{settled}, which means that its distance
value is not final. These nodes will have to be re-relaxed when their
distance values are updated. Premature relaxations are therefore
\emph{useless work}.

In our parallel implementation, each node that has to be relaxed
corresponds to a task in the scheduling system. These tasks are
prioritized using the distance value of the node, as in Dijkstra's
algorithm. For the sake of comparability to other works on
single-source shortest paths, we will use the terms \emph{node} and
\emph{relax} instead of \emph{task} and \emph{execute} throughout this
section. Instead of a priority queue, we let our scheduling system
choose the next node (task) to relax. Pseudocode for the SSSP tasks is
given in Listing~\ref{lst:sssp}.

We diverge from Dijkstra's algorithm whenever a better distance value
is found for an active node in the priority queue. Instead of updating
the priority using a \emph{decrease key} operation, we reinsert the
node into the priority queue. The previous instance of the same node,
with an old distance value as priority, is lazily removed as soon as
it is noticed. Our scheduling data structures have been implemented to
recognize such nodes lazily, and automatically remove them when
recognized. For more discussion on \emph{dead task elimination},
see~\cite{WimmerCedermanTraffTsigas13:strategies,ppopp13}.

The goal of this evaluation is to show that, using $k$-priority data
structures, the amount of useless work generated is small compared to
the actual work, and that bounds can be given on the amount of useless
work generated.

\subsection{Theoretical analysis}

For the theoretical analysis we use a simplified model of
task-parallel computations: the system operates on a global pool of
nodes (tasks), which are ordered by their tentative distance
value. Execution occurs in temporal phases and, in each phase, up to
$P$ nodes with the lowest tentative distance values are relaxed.  We
assume an \emph{ideal priority queue}, in which all nodes are visible
to all places at the beginning of each phase.  We are interested in
upper bounding the amount of useless work that is performed during
each phase.  Similar bounds have previously been obtained for
$\Delta$-stepping and other SSSP algorithms~\cite{Meyer03,MeyerSanders03}.

\subsubsection{Formal model}

We are given an undirected graph $G=(V,E)$ (with $n=|V|$ and $m=|E|$),
a source node $s\in V$ and a positive weight function $\lambda: E
\rightarrow \reals^+$. For each temporal step $t$ we maintain a
partition of $V$ into two subsets: $V=A_t \cup B_t$, of sizes
$\alpha_t$ and $\beta_t$ ($\forall t\; \alpha_t+\beta_t=n$). The set
$A_t=\{a_t(1), a_t(2), \ldots, a_t(\alpha_t)\}$ contains the active
nodes, $B_t=\{b_t(1), \ldots, b_t(\beta_t)\}$ the inactive nodes. For
each node $v\in V$ we also keep a tentative distance $\delta_t(v) \in
\reals \cup \{\infty\}$. Let $d_t(i)=\delta_t\left(a_t(i)\right)$, we
assume the nodes in $A_t$ to be ordered by $d_t$, with ties broken
arbitrarily, i.e., $\forall i\in\{1,\alpha-1\} \; d_t(i)\le
d_t(i+1)$. Initially ($t=0$) we have $A_0=\{s\}$, $B_t=V\setminus
\{s\}$, $\delta_0(s)=0$ and $\delta_0(v\not= s)=\infty$.  In each
phase (up to) $P$ active nodes $\Phi_t=\{a_t(1), \ldots, a_t(P)\}$
with lowest $d_t$ are selected and relaxed, so that at the end of the
phase the tentative distance of a generic node $w\in V$ is
\begin{equation*}
  \delta_{t+1}(w)=\min\left\{\delta_t(w),\min_{v\in\Phi_t}\left\{\delta_t(v)+\lambda(v,w)\right\}\right\} \mf
\end{equation*}
Any node (whether active or inactive) which had its tentative distance
updated is moved into $A_{t+1}$, relaxed nodes which were not updated
are moved into $B_{t+1}$, all the other nodes remain in their former
sets for the next time phase. The algorithm terminates, at some time
$\tau<n$, when there are no more active nodes, i.e.,
$A_\tau=\emptyset$ and $B_\tau=V$, with the nodes reachable from $s$
having a finite distance.

We restrict our analysis to Erd\H{o}s-R\'enyi random
graphs~\cite{ErdosR59,Gilbert59} of parameters $n$ and $p$, i.e.,
graphs with $n$ nodes, for which each of the $n \choose 2$ possible
edges has independent probability $p$ to occur. Furthermore, we
assign, independently for each edge, a weight uniformly distributed
between 0 and 1: $\forall e \in E, \; \lambda(e)\in {\cal U} ]0,1]$.
    We assume the source node $s$ to be chosen uniformly at random in
    $V$. In order to ensure, w.h.p., the connectedness of the graph,
    we also assume $p>\frac{(1+\epsilon)\ln n}{n}$ for some
    $\epsilon>0$.

\subsubsection{Useless work}

We say that a node is \emph{settled} at time $t$ when its tentative
distance is equal to its final distance. Every time that a node which
is not settled is relaxed, useless work is performed, since the node
will need to be relaxed again when its tentative distance is going to
be updated (Dijkstra's algorithm only relaxes nodes which are settled,
thus performing only useful work, but, on the other hand, it is hard
to parallelize because of its dependencies).  The following theorem
(proof in Section~\ref{sec:proofs-lemmata}) bounds the useless work
$W_t$ performed by our algorithm as a function of $d_t$.

\begin{theorem}\label{thm:useless-work}
Let $W_t$ be the useless work performed at time $t$ by our algorithm,
using an ideal priority queue, and let $h_t(i,j)=d_t(j)-d_t(i)$. We
can bound $W_t$ from above as:
\begin{equation*}
    W_t \le \sum_{j=1}^{P}\left[1- \prod_{i=1}^{j-1}\prod_{L=1}^{n-1}
      \left(1-\frac{\left(p \, h_t(i,j)\right)^L}{L!}
      \right)^{\frac{(n-2)!}{(n-1-L)!}} \right] \mf
  \end{equation*}
\end{theorem}

\begin{remark}
A simpler (but weaker) form of this bound can be obtained by
substituting $h_t(i,j)$ with $h_t^*=\max_{i,j} h_t(i,j) = h_t(1,P)$.
\end{remark}

\subsubsection{Proofs and lemmata}
\label{sec:proofs-lemmata}

In order to simplify the analysis, we assume the following property to
hold when the number of nodes $n$ is large. The property has been
experimentally validated using the simulator presented in
Section~\ref{sec:simulation}.
\begin{conjecture}
\label{cj:randomness}
    Throughout the execution of the ideal priority queue SSSP algorithm,
    for all values of $t\in\nats$, $1 \le i < j \le P$ and $h\in\, ]0,1]$,
    the probability that there is a path of weight less than $h$
    between $a_t(i)$ and $a_t(j)$ is bounded from above by the
    probability that such a path exists in a random graph, between two
    (uniformly) random nodes.
\end{conjecture}

\begin{lemma}
  Let $h\in\,]0,1]$ and let $\pi^L=\left(\pi_0, \pi_1, \ldots, \pi_{L-1},
    \pi_L\right)$ be a path in $G$
  chosen uniformly at random among the paths
  of length $L$, such that the subpaths $\pi'=(\pi_0,\ldots,\pi_{L-1})$ and
  $\pi''=(\pi_{L-1},\pi_L)$ both have weights smaller than $h$.  Let
  $f^L(\lambda)$ be the probability density function associated with the total
  weight $\lambda(\pi^L)=\sum_{i=1}^{L} \lambda(\pi_{i-1},\pi_i)$. We can write
  $f^L$ as
  \begin{equation*}
    f^L(\lambda) =
    \begin{cases}
      \frac{\lambda^{L-1}}{h^L} & \lambda \in \, ]0,h]\\
      \frac{1}{h}-\frac{(\lambda-h)^{L-1}}{h^L} & \lambda \in \, ]h,2h]\\
      0 & \text{otherwise} \mf
    \end{cases}
  \end{equation*}
\end{lemma}
\begin{proof}
  The proof is by induction on $L$.
  \par{\sc Base case.} For $L=1$, since the edge weight is uniformly distributed
  between $0$ and $1$, we clearly have
  \begin{equation*}
    f^1(\lambda) =
    \begin{cases}
      \frac{1}{h} & \lambda \in \, ]0,h]\\
      0 & \text{otherwise} \mf
    \end{cases}
  \end{equation*}
  \par{\sc Induction.} We assume now that the inductive hypothesis holds for all
  values $l\le L$. Let $f_{h}^l$ be the probability density function obtained by
  conditioning its weight $\lambda(\pi^l)$ to be smaller than $h$, i.e.,
  \begin{equation*}
    f_{h}^l(\lambda) =
    \begin{cases}
      \frac{l\lambda^{l-1}}{h^l} & \lambda \in \, ]0,h]\\
      0 & \text{otherwise} \mf
    \end{cases}
  \end{equation*}
  We have $\lambda(\pi^L)=\lambda(\pi')+\lambda(\pi'')$, where $\pi'$ and
  $\pi''$ are subpaths of length $L$ and $1$. Since $\lambda(\pi')$ and
  $\lambda(\pi'')$ are independent, the density function $f^{L+1}$ can be
  obtained by convolution of $f_{h}^L$ and $f_{h}^1$:
   \begin{align*}
     f^{L+1} &= f_{h}^L \ast f_{h}^1 &&
     \Rightarrow & f^{L+1}(\lambda) &=
    \begin{cases}
      \frac{\lambda^{L}}{h^{L+1}} & \lambda \in \, ]0,h]\\
      \frac{1}{h}-\frac{(\lambda-h)^{L}}{h^{L+1}} & \lambda \in \, ]h,2h]\\
      0 & \text{otherwise} \mc
    \end{cases}
   \end{align*}
   which concludes the induction.
\end{proof}
\begin{corollary}
  \label{cor:minpath}
  The probability that a (uniformly random) path $\pi^L$ has
  $\lambda\left(\pi^L\right) < h$, conditioned to $\lambda(\pi')<h$ and
  $\lambda(\pi'')<h$, is equal to $\frac{1}{L}$.
\end{corollary}
\begin{proof}
  Just integrate $f^L$ between $0$ and $h$.
\end{proof}

\begin{proof}[\bf Proof (Theorem~\ref{thm:useless-work})]
  Let $1\le i < j \le \alpha_t$ and let $\pi_t^L(i,j)=\left(\pi_0=a_t(i),\pi_1,
    \ldots, \pi_L=a_t(j)\right)$ be a path between $a_t(i)$ and $a_t(j)$ of
  length $L$; we denote the weight of $\pi_t^L(i,j)$ as
  \begin{equation*}
    \lambda\left(\pi_t^L(i,j)\right)=\sum_{k=0}^{L-1}
    \lambda(\pi_{k},\pi_{k+1}) \mf
  \end{equation*}
  A node $a_t(j)$ is not settled if and only if there exists $i<j$ such that
  there exists a path $\pi_t^L(i,j)$ with $\lambda\left(\pi_t^L(i,j)\right) <
  d_t(j)-d_t(i)$.  Note that the non-existence of a particular path with weight
  less than $h_t(i,j)$ does not decrease the probability for another different
  path not to exist.
  Therefore, being
  $h_t(i,j)=d_t(j)-d_t(i)\le 1$, the probability $q_t(j)$ that $a_t(j)$ is
  settled can be bounded as
  \begin{align*}
    q_t(j) &\ge \prod_{i=1}^{j-1}\prod_{L=1}^{n-1} \Pr\left[\nexists \pi_t^L(i,j) :
      \lambda\left(\pi_t^L(i,j)\right) < h_t(i,j) \right]\\
    &= \prod_{i=1}^{j-1}\prod_{L=1}^{n-1} \left(1- \rvar^L_t(i,j) \right) \mc
  \end{align*}
  where $\rvar_t^L(i,j)$ is the probability that a path $\pi_t^L(i,j)$, with
  weight less than $h_t(i,j)$, exists.  Assuming that we are relaxing the first
  $P$ nodes of $A_t$, we can compute the expected value of the useless work
  performed at time $t$ as $W_t=\sum_{j=1}^{P}\left(1-q_t(j)\right)$.  Let
  $\tilde\rvar_t^L(i,j)$ be the probability that a \emph{particular} path
  $\pi_t^L(i,j)$ exists, with weight less than $h_t(i,j)$; we can bound
  $\rvar_t^L(i,j)$ as
  \begin{equation*}
    \rvar_t^L(i,j) \le 1 - \left( 1- \tilde\rvar_t^L(i,j)\right)^{\frac{(n-2)!}{(n-1-L)!}} \mf
  \end{equation*}
  Note that there exists a path $\pi_t^L(i,j)$ with weight less than $h_t(i,j)$
  if and only if the two subpaths $\pi'=(\pi_0,\ldots,\pi_{L-1})$ and
  $\pi''=(\pi_{L-1},\pi_L)$ exist, their weights are smaller than $h_t(i,j)$,
  and so is the sum of their weights. Because of Conjecture~\ref{cj:randomness}
  and Corollary~\ref{cor:minpath} we have $\tilde\rvar_t^1(i,j)=p\, h_t(i,j)$,
  which finally implies
  \begin{align*}
    \tilde\rvar_t^L(i,j) &= \frac{\tilde\rvar_t^{L-1}(i,j) \tilde\rvar_t^1(i,j)}{L}
    = \frac{\left(\tilde\rvar_t^1(i,j)\right)^L}{L!} \mc \\ \rvar_t^L(i,j) &\le 1 -
    \left(1-\frac{\left(p \, h_t(i,j)\right)^L}{L!}
    \right)^{\frac{(n-2)!}{(n-1-L)!}} \mf
  \end{align*}
\end{proof}

\subsubsection{$k$-priority data structures}

We can adapt our theoretical framework to support $k$-priority data
structures, which allow that up to $\rho$ of the newest tasks may not
be visible to all places, and may therefore not be executed even
though they would have been with the ideal data structure. For the
centralized $k$-priority data structure $\rho=k$, for the hybrid one
$\rho=Pk$.  The bound of Theorem~\ref{thm:useless-work} can be adapted
by changing the sum over all $j$'s to only the $j$'s corresponding to
nodes $a_t(j)$ which have been actually relaxed ($\sum_{j=1}^{P}
\rightarrow \sum_{j \in R_t}$, with $R_t=\{j:a_t(j) \text{ has been
  relaxed}\}$).
Similarly to the previous case, a simpler form of this bound can be
obtained by substituting $h_t(i,j)$ with $h_t^*$, defined as the
difference between the largest and smallest tentative distance of
nodes relaxed at time $t$, which implies $h_t^* \le \max_{i,j}
h_t(i,j) = h_t(1,P+\rho)$.

\subsection{Weakening the requirements $\rho$-relaxation}
\label{sec:weak-k-order}

All our current $k$-priority data structure implementations rely on a
\emph{temporal} formulation of $\rho$-relaxation (see
Section~\ref{sec:k-relax}), allowing only the last $k$ items added to
the data structure to be ignored ($k$ items per thread for the hybrid
$k$-priority data structure). This means that after $k$ \texttt{push}
operations a thread will make all its newest tasks globally available,
regardless of how many of the $k$ new tasks have already been
executed. As it turns out, our model does not require the temporal
formulation of $\rho$-relaxation, and relies on a weaker
\emph{structural} formulation instead that requires a \texttt{pop}
operation never to ignore more than $\rho$ items regardless of their
age.

This result opens up possibilities for priority queues that achieve
similar bounds but do not need to maintain the temporal property. We
believe that this will lead to priority queues with even better
scalability than the priority queues presented in this work, and first
results with such data structures look promising.

\subsection{Simulation}
\label{sec:simulation}

\begin{figure*}[t]
\centering
% Don't remove!
%  \psfrag{hmax}[][c]{$h_t^*$}
  \includegraphics[width=0.32\textwidth]{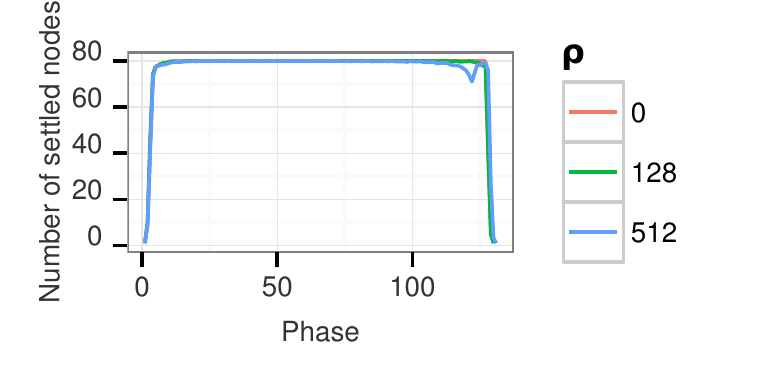}
%  \hspace{1cm}
%
  \includegraphics[width=0.32\textwidth]{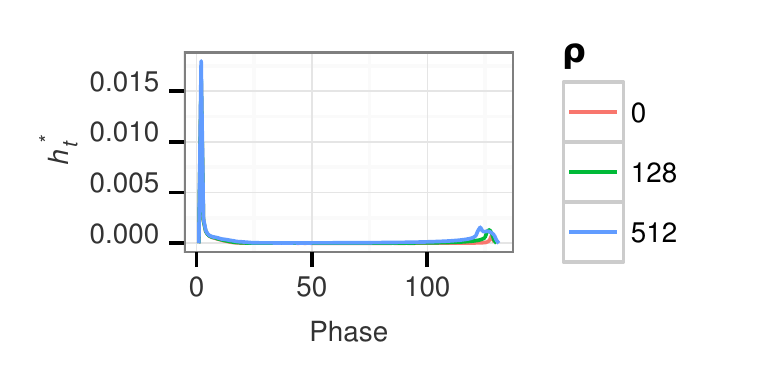}
  \includegraphics[width=0.32\textwidth]{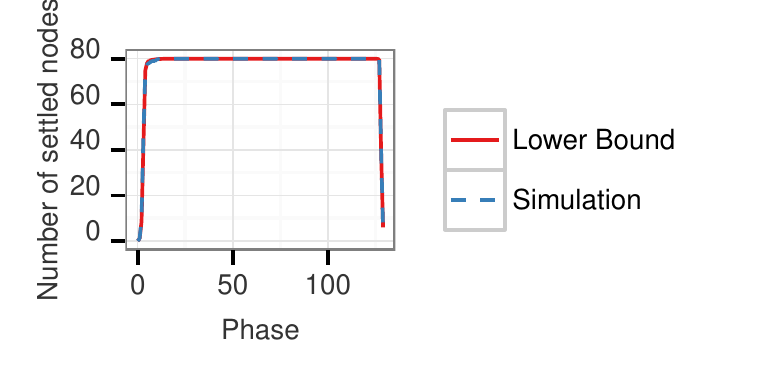}
\vspace{0.1in}
\caption{From left to right: nodes settled per phase; difference
  between biggest and smallest tentative distance of nodes relaxed per
  phase; comparison between the theoretical bound and the
  simulation. ($n=10000$, $P=80$, $p=50$\%)}
\label{fig:simt80}
\end{figure*}

We have used a simulator to bridge between the findings in our
theoretical model and the experiments in
Section~\ref{sec:experiments}.  The simulator helped understand why
$\rho$-relaxation gives such strong guarantees and was a valuable tool
to shape our theoretical analysis. The simulator uses the phase-wise
execution model used in the theoretical analysis and allows us to vary
the parameters $P$ and $\rho$. The simulator stores all active nodes
in a single array sorted by distance value. Execution proceeds in
phases, where in each phase the first $P$ nodes from the array are
relaxed. At the end of each phase the array is updated with all new
active nodes.

If $\rho>0$, newly created active nodes are marked with a sequence
id. To ensure randomness, nodes created in a single phase are shuffled
first before assigning sequence id's. The nodes with the $\rho$
highest sequence id's are stored separately from the sorted array of
nodes. These nodes represent the nodes that might be ignored due to
the $\rho$-relaxation. An exception is made if a node has the lowest
distance value of all nodes. This node is guaranteed to be relaxed in
the next phase, and is therefore added to the array of active nodes. A
deterministic tie-breaking scheme is used to ensure that only one node
has the lowest distance value of all at any time. In case that less
than $P$ nodes are available in the array, a random selection of all
other active nodes is relaxed by the other places.

\subsubsection{Simulation results}

We ran our simulator in a setting that closely resembles the setup
used in the experiments in Section~\ref{sec:experiments}. We use
exactly the same 20 random graphs used in the experiments and report
the mean. The number of places, $P$, is set to $80$, which corresponds
to the 80 cores of the machine used in our experiments. We use three
values for $\rho$: $0$, which represents an ideal priority data
structure, $128$ and $512$.

The first graph in Figure~\ref{fig:simt80} depicts the number of nodes
settled in each phase throughout the simulation. It can be seen that
for most of the execution almost all nodes that are relaxed are
already settled. Non-settled nodes are only encountered in the first
phases. For higher $\rho$ some variation can also be observed towards
the end when a significant amount of nodes is not visible to all
places. Throughout most of the execution almost all of the nodes that
are relaxed are already settled.

The middle graph in Figure~\ref{fig:simt80} shows $h_t^*$, the
difference between the largest and smallest distance value of nodes
relaxed in each phase. After only a few iterations, all of the nodes
that are relaxed have distance values close to each other, and the
distance values only grow a bit at the end of the execution, a bit
more with higher $\rho$. It is easy to see the close relationship
between distance values and nodes settled per phase.

Finally, the last graph in Figure~\ref{fig:simt80} gives a comparison
between the theoretical lower bound and the number of settled nodes in
the simulation. It can be seen that the calculated theoretical lower
bound on the number of settled nodes, and the number of nodes settled
in the simulation are very close.

\subsection{Experiments}
\label{sec:experiments}

The data structures have been evaluated on an 80-core Intel Xeon
system with 1\,TB of memory. Figure~\ref{fig:timeandtasks} shows the
average total execution time and number of spawned tasks for 20
undirected graphs, each with 10000 nodes, an edge probability of 50\%
and uniformly distributed random edge weights. The $k$ value is set to
512.

Ideally, a parallel implementation of single-source shortest paths
relaxes each node exactly once. This is the case if nodes are only
relaxed when they are settled. For our input graphs with $n=10000$,
this means that if more than $10000$ nodes were relaxed, some useless
work was performed. As can be seen in Figure~\ref{fig:timeandtasks},
close to no useless work is generated by any of the data structures,
with exception of work-stealing. With random stealing and only local
prioritization to go by, it generates more than twice the number of
necessary tasks. This shows up in the total execution time, which is
higher for work-stealing.

The parallel implementations are compared to a sequential
implementation of Dijkstra's algorithm (only shown for one
thread). Due to the small task granularity, the overhead for parallel
execution on all data structures is relatively high, but for two or
more threads the execution times drop below the sequential time. The
algorithm scales very well for up to 10 threads. For more threads the
algorithm becomes memory bandwidth bound. On the hybrid $k$-priority
data structure some more speedup can still be achieved up until 40
threads.

\begin{figure*}[t]
\includegraphics{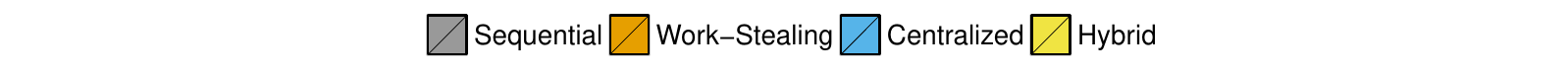}
\includegraphics{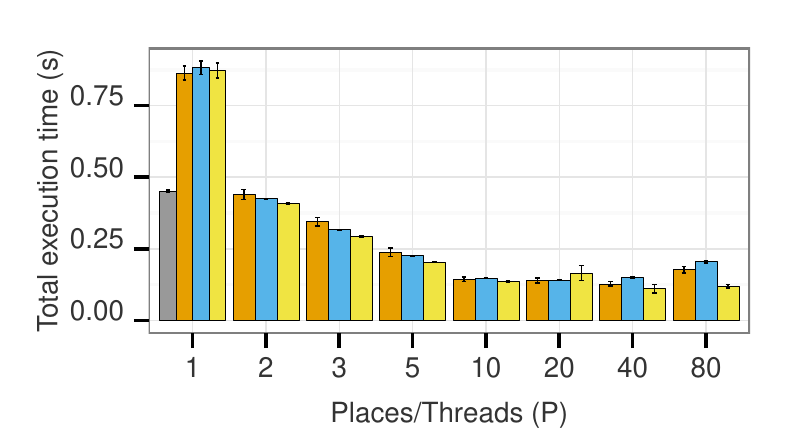}
\includegraphics{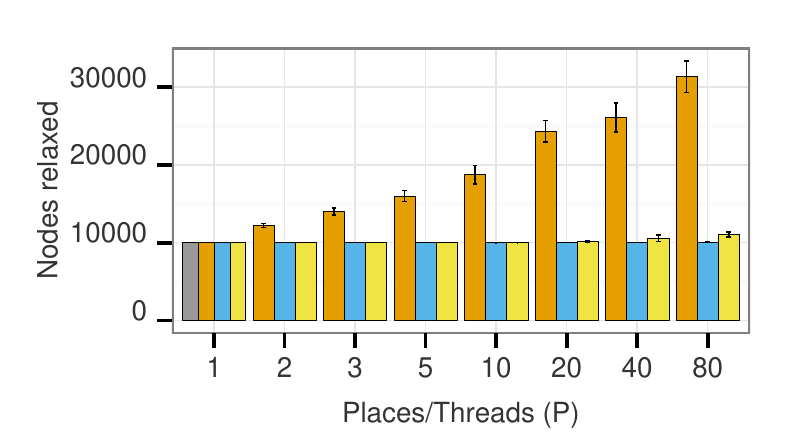}
\caption{Total execution time and number of nodes relaxed for varying
  $P$ ($n=10000$, $k=512$, $p=50$\%).}
\label{fig:timeandtasks}
\end{figure*}
\begin{figure*}[t]
\includegraphics{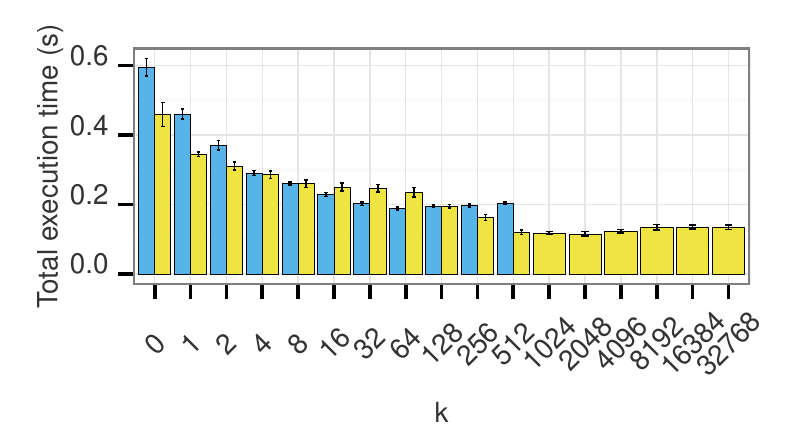}
\includegraphics{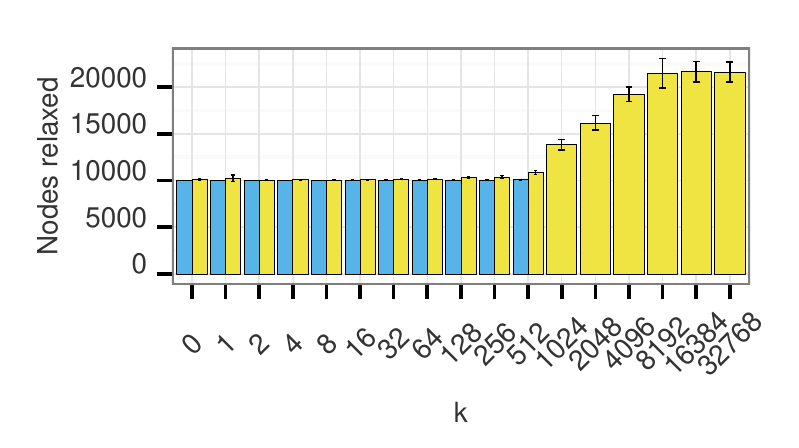}
\caption{Total execution time and number of nodes relaxed for varying
  $k$ ($n=10000$, $P=80$, $p=50\%$).}
\label{fig:timeandtasksvaryk}
\end{figure*}

Figure~\ref{fig:timeandtasksvaryk} shows, for the same graphs as in
the previous figure, the total execution time and number of spawned
tasks for different $k$ values. The number of places is fixed at
80. Here it can be seen that the centralized $k$-priority data
structure works best for $k$ in the range of $32$ to $128$. For higher
$k$ the cost of the sometimes required linear search outweighs the
gains.

The hybrid $k$-priority data structure shines with larger $k$ where it
exhibits scalability similar to work-stealing. While the wasted work
is higher than for the centralized data structure, it is still
bounded, and therefore low with the right values chosen for $k$.

The minimum $k$ required to match work-stealing performance in the
hybrid data structure is dependent on task granularity. The more
fine-grained tasks are, the higher the minimum required $k$ to match
work-stealing. It is interesting to note that even with really high
values for $k$, which result in no global synchronization, the wasted
work is still half of the wasted work in work stealing. This comes due
to the use of \emph{spying}, which allows a task to be visible to
multiple place unlike stealing, where a task is only seen by one
place.  We found $k=512$ to be the best compromise between scalability
and priority guarantees for the hybrid data structure on the given
machine.

\section{Conclusion}

We have developed lock-free data structures that can be used in task
scheduling systems to support priority scheduling. Each of these data
structures provides different trade-offs between scalability and
guarantees concerning the execution order of tasks. Our hybrid,
$k$-priority data structure allows to adjust the trade-offs using the
parameter $k$, enabling the programmer to dynamically chose between a
scalable data structure with performance comparable to work-stealing,
and a centralized data structure with strong semantic guarantees.

We evaluated all three data structures analytically, experimentally
and using a simulation. Using the single-source shortest path
algorithm as an example, we showed that, compared to work-stealing,
$\rho$-relaxation can provide a significant reduction of useless work
performed for the single-source shortest path algorithm, even with
relatively large values for $k$. Nonetheless, the limits to
scalability become visible in cases with very small task
granularities, where most of the time is spent on synchronization.

In future work we plan to explore additional data structures that
further reduce the bottlenecks while maintaining the flexibility of
the hybrid, $k$-priority data structure. We expect that $k$-relaxed
data structures that rely on the weaker, \emph{structural} formulation
of $\rho$-relaxation as described in Section~\ref{sec:weak-k-order}
will exhibit better scalability than the data structures presented in
this work, due to the reduced need for synchronization. First results
with structurally relaxed $k$-priority data structures look promising.

This work also shows how extensions to the task model, like priority
queues, can help to create simple and efficient parallel versions of
algorithms that are otherwise hard to parallelize, like Dijkstra's
single-source shortest path algorithm. In the future we plan to
provide additional scheduler data structures useful for specific
problems, allowing for an even more general use of task schedulers for
algorithm design. As an example, we plan to provide $k$-relaxed Pareto
priority queues with guarantees that can then be used for
parallelization of a multi-objective shortest path
search~\cite{Sanders13}.

\bibliographystyle{abbrv}
\bibliography{taskdatastructures}

\begin{thebibliography}{10}

\bibitem{AfekKorlandYanovsky10}
Y.~Afek, G.~Korland, and E.~Yanovsky.
\newblock Quasi-linearizability: Relaxed consistency for improved concurrency.
\newblock In {\em Principles of Distributed Systems ({OPODIS})}, volume 6490 of
  {\em {L}ecture {N}otes in {C}omputer {S}cience}, pages 395--410, 2010.

\bibitem{AhujaMagnantiOrlin93}
R.~K. Ahuja, T.~L. Magnanti, and J.~B. Orlin.
\newblock {\em Network Flows}.
\newblock Prentice-Hall, 1993.

\bibitem{Augonnet10}
C.~Augonnet, S.~Thibault, R.~Namyst, and P.-A. Wacrenier.
\newblock {StarPU}: A unified platform for task scheduling on heterogeneous
  multicore architectures.
\newblock {\em Concurrency and Computation: Practice and Experience},
  23(2):187--198, 2011.

\bibitem{BlumofeLeiserson99}
R.~D. Blumofe and C.~E. Leiserson.
\newblock Scheduling multithreaded computations by work stealing.
\newblock {\em Journal of the {ACM}}, 46(5):720--748, 1999.

\bibitem{Charles05}
P.~Charles, C.~Grothoff, V.~Saraswat, C.~Donawa, A.~Kielstra, K.~Ebcioglu,
  C.~von Praun, and V.~Sarkar.
\newblock X10: an object-oriented approach to non-uniform cluster computing.
\newblock In {\em 20th {ACM} {SIGPLAN} Conference on Object-Oriented
  Programming, Systems, Languages, and Applications {(OOPSLA)}}, pages
  519--538, 2005.

\bibitem{ErdosR59}
P.~Erd\H{o}s and A.~R\'{e}nyi.
\newblock {On random graphs}.
\newblock {\em Publicationes Mathematicae Debrecen}, 6:290--297, 1959.

\bibitem{Gilbert59}
E.~N. Gilbert.
\newblock Random graphs.
\newblock {\em The Annals of Mathematical Statistics}, 30(4):1141--1144, 1959.

\bibitem{GuoBarik09}
Y.~Guo, R.~Barik, R.~Raman, and V.~Sarkar.
\newblock Work-first and help-first scheduling policies for async-finish task
  parallelism.
\newblock In {\em 23rd {IEEE} International Parallel and Distributed Processing
  Symposium {(IPDPS)}}, pages 1--12, 2009.

\bibitem{Hendler02}
D.~Hendler and N.~Shavit.
\newblock Non-blocking steal-half work queues.
\newblock In {\em 21st Symposium on Principles of Distributed Computing
  ({PODC})}, pages 280--289. ACM, 2002.

\bibitem{kfifo}
C.~Kirsch, M.~Lippautz, and H.~Payer.
\newblock {Fast and Scalable k-FIFO Queues}.
\newblock Technical Report 2102-04, Department of Computer Sciences -
  University of Salzburg, June 2012.

\bibitem{KukanovVoss07}
A.~Kukanov and M.~J. Voss.
\newblock The foundations for scalable multi-core software in {Intel Threading
  Building Blocks}.
\newblock {\em Intel Technology Journal}, 11(4), 2007.

\bibitem{Leiserson10}
C.~E. Leiserson.
\newblock The {Cilk++} concurrency platform.
\newblock {\em The Journal of Supercomputing}, 51(3):244--257, 2010.

\bibitem{pingali2011}
A.~Lenharth, D.~Nguyen, and K.~Pingali.
\newblock Priority queues are not good concurrent priority schedulers.
\newblock Technical Report TR-11-39, The University of Texas at Austin, 2011.

\bibitem{Meyer03}
U.~Meyer.
\newblock Average-case complexity of single-source shortest-paths algorithms:
  lower and upper bounds.
\newblock {\em J. Algorithms}, 48(1):91--134, 2003.

\bibitem{MeyerSanders03}
U.~Meyer and P.~Sanders.
\newblock {$\Delta$-Stepping: A parallelizable shortest path algorithm}.
\newblock {\em Journal of Algorithms}, 49(1):114--152, 2003.

\bibitem{QuintanaOrtivandeGeijnZeeChan09}
G.~Quintana-Ort{\'i}, E.~S. Quintana-Ort{\'i}, R.~A. van~de Geijn, F.~G.~V.
  Zee, and E.~Chan.
\newblock Programming matrix algorithms-by-blocks for thread-level parallelism.
\newblock {\em ACM Transactions on Mathematical Software}, 36(3), 2009.

\bibitem{Sanders13}
P.~Sanders and L.~Mandow.
\newblock Parallel label-setting multi-objective shortest path search.
\newblock In {\em 27th IEEE International Parallel and Distributed Processing
  Symposium ({IPDPS})}, pages 215--224, 2013.

\bibitem{Wimmer13}
M.~Wimmer.
\newblock Wait-free hyperobjects for task-parallel programming systems.
\newblock In {\em 27th IEEE International Parallel and Distributed Processing
  Symposium ({IPDPS})}, 2013.

\bibitem{WimmerCedermanTraffTsigas13:strategies}
M.~Wimmer, D.~Cederman, J.~L. Tr\"aff, and P.~Tsigas.
\newblock Configurable strategies for work-stealing.
\newblock CoRR abs/1305.6474, 2013.

\bibitem{ppopp13}
M.~Wimmer, D.~Cederman, J.~L. Tr{\"a}ff, and P.~Tsigas.
\newblock Work-stealing with configurable scheduling strategies.
\newblock In {\em 18th ACM Symposium on Principles \& Practice of Parallel
  Programming {(PPoPP)}}, pages 315--316, 2013.

\bibitem{ppopp14}
M.~Wimmer, D.~Cederman, F.~Versaci, J.~L. Tr{\"a}ff, and P.~Tsigas.
\newblock Data structures for task-based priority scheduling.
\newblock In {\em 19th ACM Symposium on Principles \& Practice of Parallel
  Programming {(PPoPP)}}, page~2, 2014.

\bibitem{mtaap13}
M.~Wimmer, M.~P\"oter, and J.~L. Tr\"aff.
\newblock The {Pheet} task-scheduling framework on the {Intel}\textregistered
  {Xeon} {Phi}\texttrademark coprocessor and other multicore architectures.
\newblock In {\em Workshop on Multi-threaded Architectures and Applications
  {(MTAAP)} at the 27th International Parallel and Distributed Processing
  Symposium {(IPDPS)}}, 2013.

\end{thebibliography}

\end{document}